\documentclass[11pt]{article}

\usepackage[letterpaper,margin=1in]{geometry}
\usepackage{amsmath,amssymb,amsthm,amsfonts,bm}

\newtheorem{theorem}{Theorem}
\newtheorem{lemma}{Lemma}

\newtheorem{remark}{Remark}

\usepackage{graphicx}
\usepackage{booktabs}
\usepackage{natbib}
\usepackage[colorlinks=true,allcolors=blue]{hyperref}

\title{AIM: Amortized Inference for Multistate Transition Models}
\author{
Yuxi Zhu\textsuperscript{1,2,*},
Rui Zhang\textsuperscript{3}
\\[1em]
\small
\textsuperscript{1}Department of Pediatrics, Rainbow Babies \& Children's Hospital, Cleveland, OH, USA
\\
\small
\textsuperscript{2}School of Medicine, Case Western Reserve University, Cleveland, OH, USA
\\
\small
\textsuperscript{3}Department of Statistics, The Ohio State University, Columbus, OH, USA
\\[0.5em]
\small
\textsuperscript{*}Corresponding author: yuxi.zhu@uhhospitals.org
}
\date{}

\begin{document}
\maketitle

\begin{abstract}

Interval censored continuous time multistate transition models (MSTMs) are widely used to characterize disease progression and other dynamic processes. Existing inference procedures are predominantly likelihood-based and require a separate model fit for each dataset, involving repeated evaluation of transition probability matrices. This repeated computation can become burdensome when the same model is applied across centers, time periods, or newly collected cohorts. We propose AIM, an amortized Bayesian inference framework for interval censored continuous time MSTMs. AIM learns the relationship between informative summary statistics and model parameters from simulated datasets generated during an offline training stage. Once trained for a prespecified model class, AIM provides posterior inference for new datasets without repeated likelihood optimization.
Under fixed observation schedules and discretized covariate strata, we show that interval-specific transition counts are sufficient for the observed panel likelihood and establish identifiability of the corresponding population summaries. We further prove consistency of the induced summary posterior and its neural approximation, yielding a consistent posterior mean estimator.
Simulation studies under progressive and competing-pathway MSTMs demonstrate accurate point estimation and reliable uncertainty quantification. AIM completed online inference in milliseconds, with median speedups ranging from approximately 154-fold to 2308-fold relative to repeated likelihood-based estimation. These results establish AIM as a scalable and theoretically grounded framework for reusable Bayesian inference in interval censored MSTMs.
\end{abstract}


\section{Introduction}

MSTMs are widely used in biomedical and public health research to characterize disease progression, treatment pathways, and patient outcomes \citep{zhu2021multistate,zhu2025joint}. By modeling transition intensities between a finite set of states, they provide estimates of transition risks, state occupancy probabilities, and transition-specific covariate effects.

Inference for continuous time MSTMs is predominantly likelihood-based, using maximum likelihood estimation \citep{kalbfleisch1985analysis,jackson2011multi,gu2024maximum,zhu2024uniformization} or Bayesian posterior sampling \citep{williams2020bayesian}. For panel-observed data, these methods repeatedly evaluate transition probability matrices within numerical optimization or Markov chain Monte Carlo algorithms. More importantly, each new dataset requires a separate model fit, even when the transition structure, covariates, observation scheme, and parameter space remain unchanged. This limitation is particularly relevant in multicenter studies, disease registries, and continually updated electronic health record systems, where the same model may be repeatedly applied across centers, time periods, or newly collected cohorts. These settings motivate an approach that transfers computation to a reusable offline training stage.

Simulation-based inference (SBI) provides a framework for learning inferential procedures from simulations rather than repeated pointwise likelihood evaluation \citep{cranmer2020frontier}. Conventional approaches, including approximate Bayesian computation and synthetic likelihood, generally require a new inference procedure for each observed dataset \citep{Tavare1997,Pritchard1999,Beaumont2002,Wood2010}. Neural SBI instead learns posterior approximations from simulated parameter--dataset pairs and can therefore support amortized inference for future datasets \citep{papamakarios2016fast,greenberg2019automatic,radev2023bayesflow}. For MSTMs, \citet{tancredi2019approximate} developed an approximate Bayesian computation procedure for panel-observed continuous-time processes, but the inference remains dataset specific. To our knowledge, no existing method provides amortized neural posterior inference tailored to MSTMs. A central challenge is to construct fixed-dimensional representations that accommodate interval censoring, transition constraints, and transition-specific covariate effects while retaining the information required for posterior inference.

We propose AIM, an amortized Bayesian inference framework for MSTMs. AIM combines model-based simulation, multistate specific summary statistics, and neural posterior estimation. Datasets are simulated from a prespecified model class during an offline training stage, after which posterior inference can be performed for compatible new datasets without dataset-specific likelihood optimization or posterior sampling. Under fixed observation schedules and finitely many covariate strata, we establish sufficiency of interval-specific transition counts for the observed panel likelihood. Under appropriate reachability and design conditions, we further establish identifiability of the population summaries and consistency of the induced summary posterior, its neural approximation, and the resulting posterior mean estimator. Simulation studies under progressive and competing-pathway models evaluate estimation accuracy, uncertainty quantification, and computational efficiency relative to repeated likelihood-based analyses. A real data application illustrates the proposed framework in practice.

The remainder of the paper is organized as follows. Section~2 introduces AIM and the proposed summary representation. Section~3 presents the theoretical results. Section~4 reports the simulation studies, Section~5 presents the real-data application, and Section~6 concludes with a discussion.

\section{Method}
\label{sec:method}

\subsection{Continuous time multistate model}
\label{sec:model}

Let
\[
X_i(t)\in\{1,\ldots,S\},
\]
denote the state occupied at time \(t\) by subject \(i=1,\ldots,N\). The set of
scientifically admissible transitions is
\[
\mathcal E
=
\{(r,s):r\ne s,\ q_{rs}>0\}.
\]
Conditional on a vector of time-invariant covariates \(z_i\in\mathbb R^p\), the
process is assumed to be a time-homogeneous continuous time Markov chain with
generator
\[
Q(z_i;\theta)=\{q_{rs}(z_i;\theta)\},
\]
whose off-diagonal entries satisfy, for \((r,s)\in\mathcal E\),
\[
q_{rs}(z_i;\theta)
=
\exp\{\beta_{rs,0}+z_i^\top\beta_{rs}\},
\]
and \(q_{rs}(z_i;\theta)=0\) for transitions not in \(\mathcal E\). The diagonal
entries are \(q_{rr}(z_i;\theta)=-\sum_{s\ne r}q_{rs}(z_i;\theta)\). The unknown
parameter is
\[
\theta
=
\{(\beta_{rs,0},\beta_{rs}):(r,s)\in\mathcal E\}\in\Theta .
\]

Subjects are observed at a fixed panel of visit times
\[
0=t_0<t_1<\cdots<t_T,
\]
so that the observed data for subject \(i\) are
\[
Y_i
=
\{X_i(t_0),\ldots,X_i(t_T),z_i\}.
\]
The exact transition times and any intermediate states between two consecutive
visits are unobserved. For a subject with covariate value \(z\), the interval
transition matrix over \((t_{k-1},t_k]\) is
\[
P_k(z;\theta)=\exp\{Q(z;\theta)\Delta_k\},\qquad
\Delta_k=t_k-t_{k-1}.
\]
Thus, apart from an initial-state distribution that is treated as nuisance and
free of \(\theta\), the observed panel likelihood has the familiar product form
\[
L_N(\theta;Y)
\propto
\prod_{i=1}^N\prod_{k=1}^T
\{P_k(z_i;\theta)\}_{X_i(t_{k-1}),X_i(t_k)}.
\]
This likelihood is tractable but can be expensive to optimize repeatedly,
because each new dataset requires repeated evaluation of matrix exponentials and
a new numerical maximization. AIM is designed to move this computational cost to
an offline simulation stage.

\subsection{Informative summary statistics}
\label{sec:summary}

The neural posterior is conditioned on a fixed-dimensional representation of the
panel data. To retain covariate information while avoiding a dimension that
grows with \(N\), we partition subjects into \(G\) covariate strata. In the
implementation used here, strata are formed by median splits of selected
covariates; if \(K\) covariates are used, then \(G\le 2^K\). Let \(N_g\) denote
the number of subjects in stratum \(g\), and let \(\bar z_g\) denote the
representative covariate value. Within stratum \(g\),
\[
Q_g(\theta)=Q(\bar z_g;\theta),
\qquad
P_k^{(g)}(\theta)=\exp\{Q_g(\theta)\Delta_k\}.
\]

The first component of the summary records observed transition movement. For
each interval, stratum, origin state, and destination state, define
\[
n_{rs}^{(k,g)}
=
\sum_{i\in g}
\mathbf 1\{X_i(t_{k-1})=r,\ X_i(t_k)=s\},
\]
and use the normalized joint transition proportion
\[
\widetilde p_{rs}^{(k,g)}
=
\frac{n_{rs}^{(k,g)}}{N_g}.
\]
The second component records state occupancy:
\[
\widehat\pi_r^{(g)}(t_k)
=
\frac1{N_g}
\sum_{i\in g}\mathbf 1\{X_i(t_k)=r\},
\qquad
k=0,\ldots,T.
\]
The third component records the covariate-stratum distribution,
\[
\omega_g=\frac{N_g}{N}.
\]
The summary vector \(S_N=S(Y_1,\ldots,Y_N)\in\mathbb R^d\) is obtained by
concatenating these quantities over \(k\), \(r\), \(s\), and \(g\).

These summaries are deliberately model-facing. The transition tables encode the
observed interval likelihood; State occupancies indicate which origin states are 
present at the beginning of each interval and therefore which rows of the transition 
matrices are informed by observed transitions. and allow joint transition
proportions to be converted into conditional transition probabilities; and
stratum proportions preserve the covariate composition of the sample. The
following lemma formalizes the likelihood information carried by the transition
count component. The proof is given in Appendix~\ref{app:sufficiency_proof}.

\begin{lemma}[Sufficiency of interval-specific transition counts]
\label{lem:sufficiency_counts}
Fix the observation schedule \(0=t_0<t_1<\cdots<t_T\). Suppose that subjects are
partitioned into covariate strata \(g=1,\ldots,G\), and that within each stratum
all subjects share the representative covariate value \(\bar z_g\). Thus they
share the same generator \(Q_g(\theta)=Q(\bar z_g;\theta)\) and the same interval
transition matrices
\[
P_k^{(g)}(\theta)=\exp\{Q_g(\theta)\Delta_k\}.
\]
Assume that the initial state distribution does not depend on \(\theta\), and
that subjects are independent conditional on their stratum membership. Then
\[
T(Y)
=
\{n_{rs}^{(k,g)}: k=1,\ldots,T,\ r,s=1,\ldots,S,\ g=1,\ldots,G\}
\]
is sufficient for \(\theta\) under the discretized-covariate observed panel likelihood.
\end{lemma}

\subsection{Neural posterior approximation}
\label{sec:mdn}

Let \(\pi(\theta)\) denote the prior distribution. AIM targets the
summary-based posterior
\[
\pi(\theta\mid S_N)\propto p(S_N\mid\theta)\pi(\theta),
\]
without evaluating the generally intractable summary likelihood
\(p(S_N\mid\theta)\). The conditional density is approximated using a mixture
density network (MDN). For an input summary
\[
x=S_N,
\]
the network outputs mixture weights \(w_1(x),\ldots,w_H(x)\), component means
\(\mu_1(x),\ldots,\mu_H(x)\), and positive diagonal covariance matrices
\(\Sigma_1(x),\ldots,\Sigma_H(x)\). The approximate posterior is
\[
q_\phi(\theta\mid x)
=
\sum_{h=1}^H
w_h(x)\,
\mathcal N\{\theta;\mu_h(x),\Sigma_h(x)\},
\]
where \(w_h(x)\ge0\) and \(\sum_h w_h(x)=1\). In practice, the network enforces
these constraints through a softmax layer for the weights and log-scale outputs
for the diagonal standard deviations.

Training pairs are generated from the prior predictive distribution:
\[
\theta^{(m)}\sim\pi(\theta),
\qquad
Y^{(m)}\sim p(Y\mid\theta^{(m)}),
\qquad
x^{(m)}=S(Y^{(m)}),
\qquad
m=1,\ldots,M.
\]
The MDN is trained by minimizing the empirical negative conditional
log-likelihood
\[
L_M(\phi)
=
-\frac1M
\sum_{m=1}^M
\log q_\phi\{\theta^{(m)}\mid x^{(m)}\}.
\]
This objective is the standard proper scoring rule for conditional density
estimation: at the population optimum, the best possible network recovers the
posterior distribution of \(\theta\) given the summary.

\subsection{Offline training and posterior inference}
\label{sec:offline_online}

The defining feature of AIM is the separation between a computationally
intensive offline stage and a lightweight online stage. Offline, one simulates
many complete datasets from the chosen model class, converts each dataset into
the summary vector, and fits the MDN. Online, a new observed dataset is reduced
to the same summary vector and passed once through the trained network.

Given posterior samples
\[
\theta^{(1)},\ldots,\theta^{(B)}
\sim q_{\hat\phi}(\theta\mid S_N),
\]
we report the posterior mean
\[
\hat\theta_{\mathrm{mean}}
=
\frac1B\sum_{b=1}^B\theta^{(b)},
\]

Posterior standard deviations and credible intervals are computed from the
Monte Carlo draws:
\[
\widehat{\operatorname{sd}}(\theta)
=
\left[
\frac1{B-1}
\sum_{b=1}^B
\{\theta^{(b)}-\hat\theta_{\mathrm{mean}}\}^2
\right]^{1/2},
\]
with a \((1-\alpha)\times100\%\) credible interval given by the empirical
\(\alpha/2\) and \(1-\alpha/2\) quantiles.

The online computational cost is essentially the cost of constructing
\(S_N\), one forward pass through the MDN, and sampling from a finite Gaussian
mixture. For fixed \(S\), \(T\), and \(G\), summary construction is \(O(NT)\),
whereas the neural evaluation and posterior sampling are independent of the
optimization burden that appears in likelihood-based fitting. The offline cost
can be substantial, but it is paid once for a specified model class and can then
be amortized across repeated datasets, subgroups, resamples, or external
applications.

\begin{center}
\begin{minipage}{0.94\textwidth}
\hrule
\vspace{0.6em}
\noindent\textbf{Algorithm 1. AIM: Offline Training and Posterior Inference}
\vspace{0.4em}

\noindent\textbf{Inputs:} transition graph \(\mathcal E\), observation schedule
\(\{t_k\}_{k=0}^T\), prior \(\pi(\theta)\), simulator \(p(Y\mid\theta)\),
summary map \(S(\cdot)\), number of simulations \(M\), and number of posterior
draws \(B\).

\begin{enumerate}
\item \textbf{Offline simulation.} For \(m=1,\ldots,M\), draw
\(\theta^{(m)}\sim\pi(\theta)\), simulate \(Y^{(m)}\sim p(Y\mid\theta^{(m)})\),
and compute \(x^{(m)}=S(Y^{(m)})\).
\item \textbf{Offline training.} Fit the MDN by minimizing
\[
-M^{-1}\sum_{m=1}^M\log q_\phi\{\theta^{(m)}\mid x^{(m)}\}.
\]
\item \textbf{Online summarization.} For a new dataset \(Y_{\mathrm{obs}}\),
compute \(x_{\mathrm{obs}}=S(Y_{\mathrm{obs}})\).
\item \textbf{Online posterior inference.} Draw
\(\theta^{(1)},\ldots,\theta^{(B)}\sim
q_{\hat\phi}(\theta\mid x_{\mathrm{obs}})\), and report posterior means,
standard deviations, and credible intervals.
\end{enumerate}
\vspace{0.2em}
\hrule
\end{minipage}
\end{center}

\section{Theoretical Properties}
\label{sec:theory}

This section gives the statistical justification for the amortized procedure.
There are two distinct approximation steps. First, the full panel \(Y_N\) is
replaced by the fixed-dimensional summary \(S_N\). Second, the summary posterior
\(\Pi(\cdot\mid S_N)\) is approximated by the learned neural density
 \(q_{\hat\psi_M}(\cdot\mid S_N)\). We show that the summaries identify the
model parameters for the transition structures considered here, that the
summary posterior is consistent, and that a sufficiently accurate neural
approximation inherits this consistency.

Let \(\Theta\) denote the parameter space. Let \(S_N=S(Y_1,\ldots,Y_N)\) be the
summary vector computed from \(N\) observed subjects, and let
\(\Pi(\cdot\mid S_N)\) denote the posterior distribution induced by the prior
\(\pi(\theta)\) and the summary likelihood.

\subsection{Population summary identifiability}
\label{sec:population_identifiability}

We first define the population version of the summary vector. For stratum \(g\), let
\[
\pi_r^{(g)}(t_k;\theta)
=
\Pr_\theta\{X_i(t_k)=r\mid i\in g\}
\]
denote the marginal state probability at visit time \(t_k\). The population joint transition proportion is
\[
\widetilde p_{rs}^{(k,g)}(\theta)
=
\pi_r^{(g)}(t_{k-1};\theta)
\{P_k^{(g)}(\theta)\}_{rs}.
\]
Let
\[
S^*(\theta)
=
\left\{
\widetilde p_{rs}^{(k,g)}(\theta),\
\pi_r^{(g)}(t_k;\theta),\
\omega_g
:\ k=1,\ldots,T,\ k=0,\ldots,T,\ r,s=1,\ldots,S,\ g=1,\ldots,G
\right\}
\]
denote the population summary vector. The stratum proportions \(\omega_g\) are structural constants determined by the covariate stratification and do not depend on \(\theta\).

The following lemma shows that, under reachability and full-rank design conditions, the population summaries identify the model parameters.

\begin{lemma}[Population summary identifiability]
\label{lem:population_identifiability}
Suppose that the following conditions hold.
\begin{enumerate}
\item[(I1)] \emph{Reachability.} For every parameterized transition \((r,s)\in\mathcal E\), there exists at least one pair \((k,g)\) such that
\[
\pi_r^{(g)}(t_{k-1};\theta)>0
\]
for all \(\theta\in\Theta\).

\item[(I2)] \emph{Generator identifiability from interval transition matrices.} For the transition graphs considered, the collection of interval transition matrices \(\{P_k^{(g)}(\theta):k=1,\ldots,T\}\) uniquely determines \(Q_g(\theta)\) for each stratum \(g\). This condition holds for the acyclic progressive and competing-pathway structures considered in the simulations.

\item[(I3)] \emph{Full-rank stratum design.} Let
\[
Z^*
=
\begin{pmatrix}
1 & \bar z_1^\top\\
\vdots & \vdots\\
1 & \bar z_G^\top
\end{pmatrix}.
\]
Then
\[
\operatorname{rank}(Z^*)=p+1.
\]
\end{enumerate}
Then the map \(\theta\mapsto S^*(\theta)\) is injective. That is,
\[
S^*(\theta_1)=S^*(\theta_2)
\quad\Longrightarrow\quad
\theta_1=\theta_2.
\]
\end{lemma}

\begin{proof}
Suppose that \(S^*(\theta_1)=S^*(\theta_2)\). We prove that \(\theta_1=\theta_2\).

First, fix \((k,g,r)\) such that \(\pi_r^{(g)}(t_{k-1};\theta_j)>0\) for \(j=1,2\). Since both the population joint transition proportions and the marginal state proportions are included in \(S^*(\theta)\), equality of the population summaries implies
\[
\widetilde p_{rs}^{(k,g)}(\theta_1)=\widetilde p_{rs}^{(k,g)}(\theta_2),
\qquad
\pi_r^{(g)}(t_{k-1};\theta_1)=\pi_r^{(g)}(t_{k-1};\theta_2)
\]
for all destination states \(s\). Therefore,
\[
\{P_k^{(g)}(\theta_1)\}_{rs}
=
\frac{\widetilde p_{rs}^{(k,g)}(\theta_1)}{\pi_r^{(g)}(t_{k-1};\theta_1)}
=
\frac{\widetilde p_{rs}^{(k,g)}(\theta_2)}{\pi_r^{(g)}(t_{k-1};\theta_2)}
=
\{P_k^{(g)}(\theta_2)\}_{rs}.
\]
Thus the conditional transition probabilities are recovered on all reachable rows. By the reachability condition, every parameterized transition originates from a state that is reachable in at least one interval and stratum. Hence the rows of the interval transition matrices relevant to the parameterized transitions are identified from the population summaries.

By (I2), the collection of recovered interval transition matrices uniquely determines the stratum-specific generator. Therefore,
\[
Q_g(\theta_1)=Q_g(\theta_2)
\]
for every stratum \(g\). In particular, for each allowed transition \((r,s)\in\mathcal E\),
\[
q_{rs}(\bar z_g;\theta_1)=q_{rs}(\bar z_g;\theta_2),
\qquad
 g=1,\ldots,G.
\]
Taking logarithms gives
\[
\beta_{rs,0}^{(1)}+\bar z_g^\top\beta_{rs}^{(1)}
=
\beta_{rs,0}^{(2)}+\bar z_g^\top\beta_{rs}^{(2)},
\qquad
 g=1,\ldots,G.
\]
Equivalently,
\[
Z^*
\begin{pmatrix}
\beta_{rs,0}^{(1)}-\beta_{rs,0}^{(2)}\\
\beta_{rs}^{(1)}-\beta_{rs}^{(2)}
\end{pmatrix}
=0.
\]
Since \(Z^*\) has full column rank by (I3), the vector in parentheses is zero. Thus
\[
\beta_{rs,0}^{(1)}=\beta_{rs,0}^{(2)},
\qquad
\beta_{rs}^{(1)}=\beta_{rs}^{(2)}.
\]
Repeating this argument for every \((r,s)\in\mathcal E\) gives \(\theta_1=\theta_2\).
\end{proof}

\begin{remark}
The inclusion of marginal state proportions in the summary vector is important for identifiability. The empirical transition quantities used in the code are joint proportions of the form \(n_{rs}^{(k,g)}/N_g\). The marginal proportions allow these joint quantities to be normalized into conditional transition probabilities whenever the origin state is reachable.
\end{remark}

\begin{remark}
The generator-identifiability condition in Lemma~\ref{lem:population_identifiability} is natural for the simulation models considered here. In the three-state progressive model and the four-state competing-pathway model, the transition graph is acyclic, and after a suitable state ordering the generator is triangular. The diagonal entries of \(P_k^{(g)}(\theta)\) determine exit-rate combinations, while off-diagonal transition probabilities determine the allocation of exits among allowed transitions. Thus distinct parameter values induce distinct collections of interval transition matrices.
\end{remark}

\subsection{Summary posterior consistency}
\label{sec:summary_posterior_consistency}

We next show that the summary-based posterior concentrates around the true parameter value. Let \(\theta_0\in\Theta\) denote the data-generating parameter and define
\[
A_\varepsilon
=
\{\theta\in\Theta:\|\theta-\theta_0\|\ge\varepsilon\}.
\]

\begin{theorem}[Summary posterior consistency]
\label{thm:summary_posterior_consistency}
Suppose that the following conditions hold.
\begin{enumerate}
\item[(C1)] The parameter space \(\Theta\) is compact and \(\theta_0\in\operatorname{int}(\Theta)\).
\item[(C2)] The prior density \(\pi(\theta)\) is continuous and strictly positive in a neighborhood of \(\theta_0\).
\item[(C3)] The fixed-schedule and discretized-covariate likelihood is given by the product-multinomial likelihood induced by the interval-specific transition counts in Lemma~\ref{lem:sufficiency_counts}.
\item[(C4)] The transition probability vectors have common support and are continuous in \(\theta\), with positive probabilities bounded away from zero on their common non-structural support.
\item[(C5)] The population summary map is identifiable as in Lemma~\ref{lem:population_identifiability}.
\item[(C6)] The limiting proportions of subjects in rows used for identification are nondegenerate. That is,
\[
\frac{n_{r\cdot}^{(k,g)}}{N}
\stackrel p\longrightarrow
\rho_r^{(k,g)}(\theta_0)>0
\]
for the relevant rows.
\end{enumerate}
Then, for every \(\varepsilon>0\),
\[
\Pi(A_\varepsilon\mid S_N)
\stackrel p\longrightarrow
0.
\]
\end{theorem}

\begin{proof}
By Lemma~\ref{lem:sufficiency_counts}, the observed panel likelihood can be written, up to a factor free of \(\theta\), as
\[
L_N(\theta)
=
\prod_{g=1}^G
\prod_{k=1}^T
\prod_{r=1}^S
\prod_{s=1}^S
\{p_{rs}^{(k,g)}(\theta)\}^{n_{rs}^{(k,g)}}.
\]
Define the normalized log-likelihood
\[
\ell_N(\theta)
=
\frac1N\log L_N(\theta)
=
\sum_{g,k,r,s}
\frac{n_{rs}^{(k,g)}}{N}
\log p_{rs}^{(k,g)}(\theta).
\]
Under \(P_{\theta_0}\), the law of large numbers gives
\[
\frac{n_{rs}^{(k,g)}}{N}
\stackrel p\longrightarrow
\rho_r^{(k,g)}(\theta_0)
 p_{rs}^{(k,g)}(\theta_0).
\]
Hence \(\ell_N(\theta)\) converges uniformly in probability to
\[
\ell(\theta)
=
\sum_{g,k,r}
\rho_r^{(k,g)}(\theta_0)
\sum_s
p_{rs}^{(k,g)}(\theta_0)
\log p_{rs}^{(k,g)}(\theta),
\]
where uniformity follows from compactness, continuity, and the bounded log-probability condition.

For any \(\theta\in\Theta\),
\[
\ell(\theta)-\ell(\theta_0)
=
-
\sum_{g,k,r}
\rho_r^{(k,g)}(\theta_0)
KL\{p_r^{(k,g)}(\theta_0)\|p_r^{(k,g)}(\theta)\}
\le 0.
\]
By Lemma~\ref{lem:population_identifiability}, if \(\theta\ne\theta_0\), at least one relevant transition probability vector differs, and hence at least one KL divergence is strictly positive. Thus \(\ell(\theta)<\ell(\theta_0)\) for \(\theta\ne\theta_0\). Since \(A_\varepsilon\) is compact, there exists \(c_\varepsilon>0\) such that
\[
\sup_{\theta\in A_\varepsilon}
\{\ell(\theta)-\ell(\theta_0)\}
\le
-c_\varepsilon.
\]
Uniform convergence then implies that, with probability tending to one,
\[
\sup_{\theta\in A_\varepsilon}
\{\ell_N(\theta)-\ell_N(\theta_0)\}
\le
-\frac{c_\varepsilon}{2}.
\]
Equivalently,
\[
\sup_{\theta\in A_\varepsilon}
\frac{L_N(\theta)}{L_N(\theta_0)}
\le
\exp\left(-\frac{Nc_\varepsilon}{2}\right).
\]

The posterior mass of \(A_\varepsilon\) is
\[
\Pi(A_\varepsilon\mid S_N)
=
\frac{\int_{A_\varepsilon}L_N(\theta)\pi(\theta)d\theta}
{\int_\Theta L_N(\theta)\pi(\theta)d\theta}.
\]
The numerator is bounded by
\[
L_N(\theta_0)
\exp\left(-\frac{Nc_\varepsilon}{2}\right).
\]
For the denominator, by positivity of the prior near \(\theta_0\) and continuity of \(\ell(\theta)\), there exists a ball \(B_\delta(\theta_0)\) and a constant \(C_\delta>0\) such that, with probability tending to one,
\[
\int_\Theta L_N(\theta)\pi(\theta)d\theta
\ge
C_\delta L_N(\theta_0)
\exp\left(-\frac{Nc_\varepsilon}{4}\right).
\]
Combining the numerator and denominator bounds yields
\[
\Pi(A_\varepsilon\mid S_N)
\le
C_\delta^{-1}
\exp\left(-\frac{Nc_\varepsilon}{4}\right),
\]
with probability tending to one. The right-hand side converges to zero, proving the result.
\end{proof}

\subsection{Simulation consistency of the neural posterior}
\label{sec:simulation_consistency}

Let
\[
(\theta_i,S_i),\qquad i=1,\ldots,M,
\]
be simulated training pairs generated from
\[
p(\theta,S)=\pi(\theta)p(S\mid\theta).
\]
The MDN is trained by minimizing
\[
L_M(\psi)
=
-\frac1M\sum_{i=1}^M
\log q_\psi(\theta_i\mid S_i).
\]
Define
\[
L(\psi)=E_{(\theta,S)}[-\log q_\psi(\theta\mid S)].
\]

\begin{theorem}[Simulation consistency]
\label{thm:sim_consistency}
Suppose that \(\Psi\) is compact, \(q_\psi(\theta\mid S)\) is continuous and strictly positive in \(\psi\), and the loss is dominated by an integrable envelope. Suppose further that the neural density class is rich enough to contain the summary-based posterior in its closure. Let
\[
\hat\psi_M\in\arg\min_{\psi\in\Psi}L_M(\psi).
\]
Then
\[
E_S
\left[
KL\{\Pi(\cdot\mid S)\|q_{\hat\psi_M}(\cdot\mid S)\}
\right]
\stackrel p\longrightarrow 0.
\]
\end{theorem}

\begin{proof}
By the joint factorization \(p(\theta,S)=\pi(\theta\mid S)p(S)\),
\[
L(\psi)
=
\int p(S)
\left[\int -\log q_\psi(\theta\mid S)\pi(\theta\mid S)d\theta\right]dS.
\]
Adding and subtracting the posterior entropy gives
\[
L(\psi)
=
H+E_S[KL\{\Pi(\cdot\mid S)\|q_\psi(\cdot\mid S)\}],
\]
where \(H=-E[\log\pi(\theta\mid S)]\) does not depend on \(\psi\). Under the approximation assumption, the infimum of \(L(\psi)\) equals \(H\).

The compactness, continuity, and envelope assumptions imply a uniform law of large numbers for \(L_M(\psi)\). Standard argmin consistency therefore gives
\[
L(\hat\psi_M)-\inf_{\psi\in\Psi}L(\psi)
\stackrel p\longrightarrow0.
\]
Using the preceding decomposition, the claimed expected KL convergence follows.
\end{proof}

\subsection{Neural posterior consistency}
\label{sec:neural_posterior_consistency}

We now combine summary posterior consistency with neural posterior approximation.

\begin{theorem}[Neural posterior consistency]
\label{thm:neural_posterior_consistency}
Suppose that the assumptions of Theorem~\ref{thm:summary_posterior_consistency} hold. Suppose further that, along the observed summary sequence,
\[
KL\{\Pi(\cdot\mid S_N)\|q_{\hat\psi_M}(\cdot\mid S_N)\}
\stackrel p\longrightarrow0
\]
as \(M\to\infty\). Then, for every \(\varepsilon>0\),
\[
q_{\hat\psi_M}(\|\theta-\theta_0\|\ge\varepsilon\mid S_N)
\stackrel p\longrightarrow0,
\]
as \(N\to\infty\) and \(M\to\infty\). Equivalently,
\[
q_{\hat\psi_M}(\cdot\mid S_N)
\Rightarrow
\delta_{\theta_0}
\]
in probability.
\end{theorem}

\begin{proof}
Fix \(\varepsilon>0\) and let \(A_\varepsilon=\{\theta:\|\theta-\theta_0\|\ge\varepsilon\}\). By Pinsker's inequality,
\[
\|q_{\hat\psi_M}(\cdot\mid S_N)-\Pi(\cdot\mid S_N)\|_{\mathrm{TV}}
\le
\left[
\frac12
KL\{\Pi(\cdot\mid S_N)\|q_{\hat\psi_M}(\cdot\mid S_N)\}
\right]^{1/2}.
\]
The right-hand side converges to zero in probability. Hence
\[
|q_{\hat\psi_M}(A_\varepsilon\mid S_N)-\Pi(A_\varepsilon\mid S_N)|
\le
\|q_{\hat\psi_M}(\cdot\mid S_N)-\Pi(\cdot\mid S_N)\|_{\mathrm{TV}}
=o_p(1).
\]
By Theorem~\ref{thm:summary_posterior_consistency},
\[
\Pi(A_\varepsilon\mid S_N)\stackrel p\longrightarrow0.
\]
Therefore
\[
q_{\hat\psi_M}(A_\varepsilon\mid S_N)
=\Pi(A_\varepsilon\mid S_N)+o_p(1)
\stackrel p\longrightarrow0.
\]
This proves concentration on every neighborhood of \(\theta_0\), which is equivalent to weak convergence in probability to \(\delta_{\theta_0}\).
\end{proof}

\subsection{Posterior mean consistency}
\label{sec:posterior_mean_consistency}

\begin{theorem}[Posterior mean consistency]
\label{thm:posterior_mean_consistency}
Under the assumptions of Theorem~\ref{thm:neural_posterior_consistency}, define
\[
\hat\theta_N
=
\int_\Theta \theta\,q_{\hat\psi_M}(\theta\mid S_N)d\theta.
\]
Then
\[
\hat\theta_N\stackrel p\longrightarrow\theta_0.
\]
\end{theorem}

\begin{proof}
Since \(\Theta\) is compact, the identity map \(f(\theta)=\theta\) is bounded and continuous. By Theorem~\ref{thm:neural_posterior_consistency},
\[
q_{\hat\psi_M}(\cdot\mid S_N)
\Rightarrow
\delta_{\theta_0}
\]
in probability. Therefore,
\[
\int_\Theta \theta\,q_{\hat\psi_M}(\theta\mid S_N)d\theta
\longrightarrow
\int_\Theta \theta\,d\delta_{\theta_0}
=
\theta_0,
\]
in probability.
\end{proof}

\section{Simulation studies}
\label{sec:simulation}

We conducted simulation studies to evaluate whether AIM can recover the
transition-intensity parameters of continuous time multistate models while
delivering the intended computational gain after offline training. The
experiments were designed to separate three issues: accuracy of posterior point
estimation, calibration of uncertainty intervals, and online cost relative to a
standard likelihood-based implementation. The \texttt{msm} package was used as
the benchmark because it fits the same continuous-time Markov models by
dataset-specific likelihood optimization.

\subsection{Simulation settings}

Three scenarios were considered, increasing in dimension and transition-graph
complexity.

\paragraph{Scenario A: Three-state progressive model without covariates.}

The baseline setting was a progressive three-state model with transition
structure

\[
1\rightarrow2,\qquad
1\rightarrow3,\qquad
2\rightarrow3.
\]

The transition intensities were parameterized as

\[
q_{rs}
=
\exp(\beta_{rs,0}),
\]

with true parameter values

\[
\beta_{12,0}=-0.6,\qquad
\beta_{13,0}=-1.0,\qquad
\beta_{23,0}=-0.2.
\]

\paragraph{Scenario B: Three-state model with covariates.}

The second setting used the same transition graph but included two independent
covariates

\[
Z_1,Z_2\sim N(0,1)
\]

were incorporated through a proportional intensity model,

\[
q_{rs}(z)
=
\exp\left(
\beta_{rs,0}
+
\beta_{rs,1}Z_1
+
\beta_{rs,2}Z_2
\right).
\]

The true parameter vector was

\[
(-0.6,0.5,-0.3,
-1.0,0.4,0.2,
-0.2,-0.5,0.6).
\]

\paragraph{Scenario C: Four-state competing pathway model.}

The third setting was a four-state competing-pathway model with transition
structure

\[
1\rightarrow2,\qquad
1\rightarrow3,\qquad
2\rightarrow4,\qquad
3\rightarrow4,
\]

together with two covariates. Transition intensities again followed the proportional intensity model above. The corresponding true parameter vector was

\[
(-0.8,0.5,-0.3,
-1.0,0.4,0.2,
-0.4,-0.5,0.6,
-0.3,0.2,0.5).
\]

For all scenarios, data from \(N=5000\) subjects were generated and observed
only at the fixed visit times

\[
(0,\;0.5,\;1,\;1.5,\;2),
\]

resulting in interval-censored panel data. The AIM and \texttt{msm} analyses
used the same simulated panels and the same transition structures.

\subsection{Training procedure}

For each scenario, AIM was trained once using synthetic datasets generated from
the prior predictive distribution,

\[
\theta^{(m)}
\sim
\pi(\theta),
\]

followed by

\[
X^{(m)}
\sim
p(X\mid\theta^{(m)}),
\]

for

\[
m=1,\ldots,50000.
\]

The summaries in Section~\ref{sec:method} were computed for each simulated
dataset and used as MDN inputs. The network consisted of two hidden layers with
256 neurons and twelve Gaussian mixture components. Parameters were estimated
by minimizing the negative conditional log-likelihood using the Adam optimizer
with learning rate \(10^{-3}\). After this offline stage, the trained network
was fixed and applied unchanged to all test datasets from the corresponding
scenario.

\subsection{Performance measures}

For each scenario, 100 independent test datasets were generated from the true
parameter values. For each test dataset, posterior samples were drawn from AIM
and likelihood-based estimates and Wald intervals were obtained from
\texttt{msm}. We report, parameter by parameter,

\begin{enumerate}
\item the average point estimate;
\item the average posterior or sampling standard deviation;
\item empirical coverage probability of the nominal \(95\%\) interval.
\end{enumerate}

For AIM, intervals are posterior credible intervals; for \texttt{msm}, they are
Wald confidence intervals based on the fitted likelihood. Online runtime was
measured after AIM training, so it reflects only the cost required for a new
dataset.

\subsection{Results}

Tables~\ref{tab:simulation_3state} and~\ref{tab:simulation_4state} summarize
the inferential results. In all three scenarios, AIM posterior means tracked the
true parameters closely and were comparable to the likelihood-based
\texttt{msm} estimates. In the three-state model without covariates, AIM
estimated the three log-intensities as \(-0.605\), \(-0.986\), and \(-0.199\),
compared with true values \(-0.600\), \(-1.000\), and \(-0.200\). The
corresponding \texttt{msm} estimates were \(-0.601\), \(-1.000\), and
\(-0.200\). Coverage was near nominal for all three parameters.

In the covariate-adjusted three-state model, AIM recovered both intercepts and
covariate effects with small bias. The six covariate effects
\((0.500,-0.300,0.400,0.200,-0.500,0.600)\) were estimated as
\((0.472,-0.315,0.409,0.199,-0.514,0.584)\). The \texttt{msm} estimates were
also close to the truth and had coverage closer to 0.95. AIM coverage ranged
from 0.850 to 0.970 in this setting.

The four-state competing-pathway model provides the strongest stress test,
because the parameter dimension increases to twelve and the transition graph
contains two competing routes to the terminal state. AIM remained stable:
posterior means were close to the true values for the competing transitions
\(1\to2\) and \(1\to3\) and for the downstream transitions \(2\to4\) and
\(3\to4\). Coverage ranged from 0.870 to 0.990 for AIM and from 0.910 to 0.980
for \texttt{msm}. Thus, the amortized estimator preserved good point estimation
performance as model complexity increased, with some loss of interval
calibration relative to direct likelihood fitting.

Figure~\ref{fig:simulation} shows that AIM posterior mean estimates closely follow the true parameter values across all three simulation scenarios and are broadly comparable to the likelihood-based \texttt{msm} estimates. Empirical coverage probabilities are generally near the nominal 95\% level, with modest parameter-specific variation in the covariate and four-state settings.

\begin{figure}[htbp]
\centering
\includegraphics[width=\textwidth]{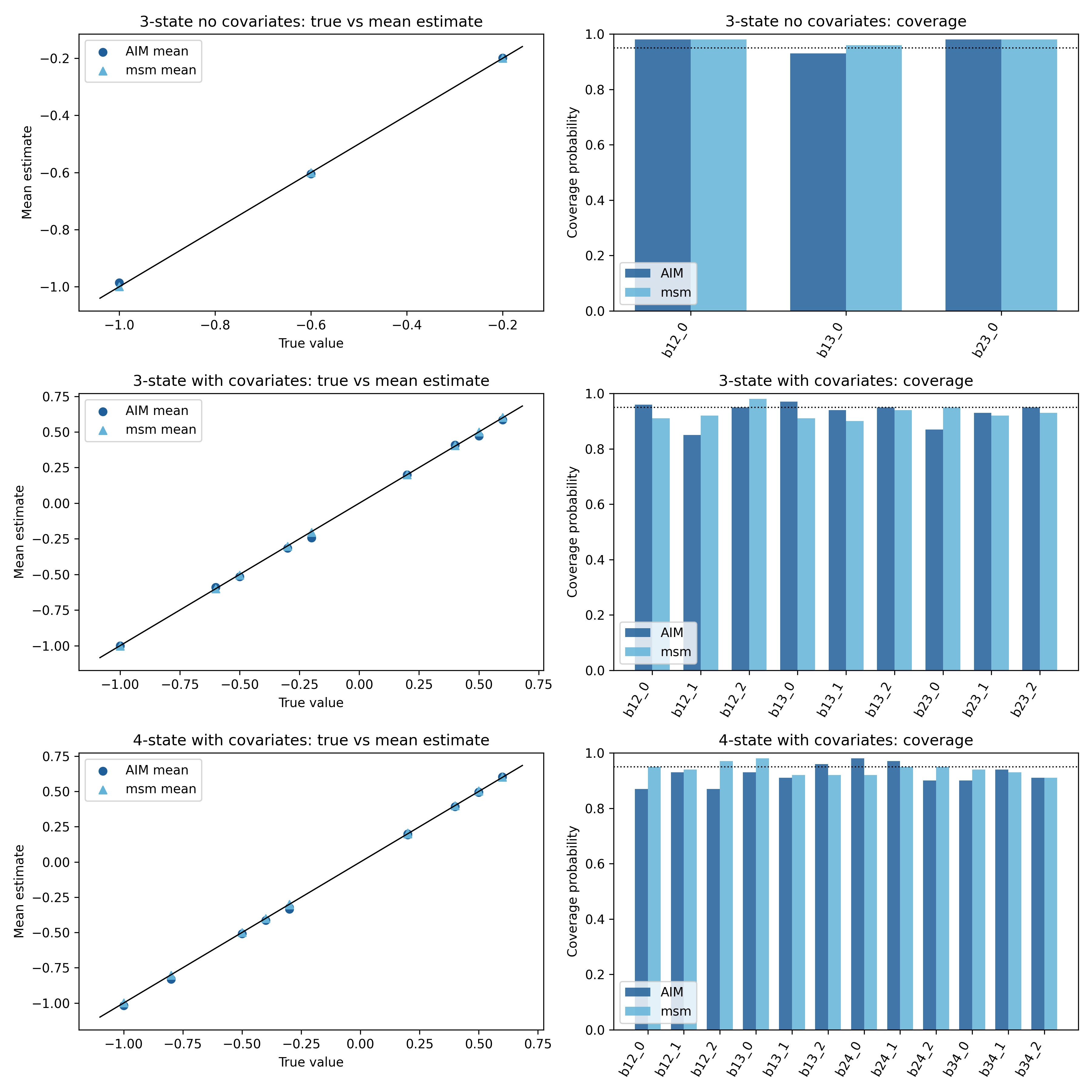}
\caption{
Simulation performance of AIM and the likelihood-based \texttt{msm} estimator across the three data-generating scenarios. Left panels compare the average point estimates with the true parameter values; the diagonal line indicates perfect agreement. Right panels show empirical coverage probabilities of nominal 95\% intervals, with the dotted horizontal line indicating the nominal level. AIM posterior mean estimates closely track the true values and are broadly comparable to \texttt{msm}, with coverage generally near the nominal level across scenarios.
}
\label{fig:simulation}
\end{figure}

As shown in Figure~\ref{fig:online_time}, AIM substantially reduced online inference time relative to \texttt{msm}. Across replicated datasets, AIM performed inference in milliseconds, while \texttt{msm} required repeated likelihood optimization. Median speedups were approximately 154-fold, 964-fold, and 2308-fold across the three simulation scenarios, with the largest gain observed in the four-state competing-pathway model.

\begin{figure}[htbp]
\centering
\includegraphics[width=\textwidth]{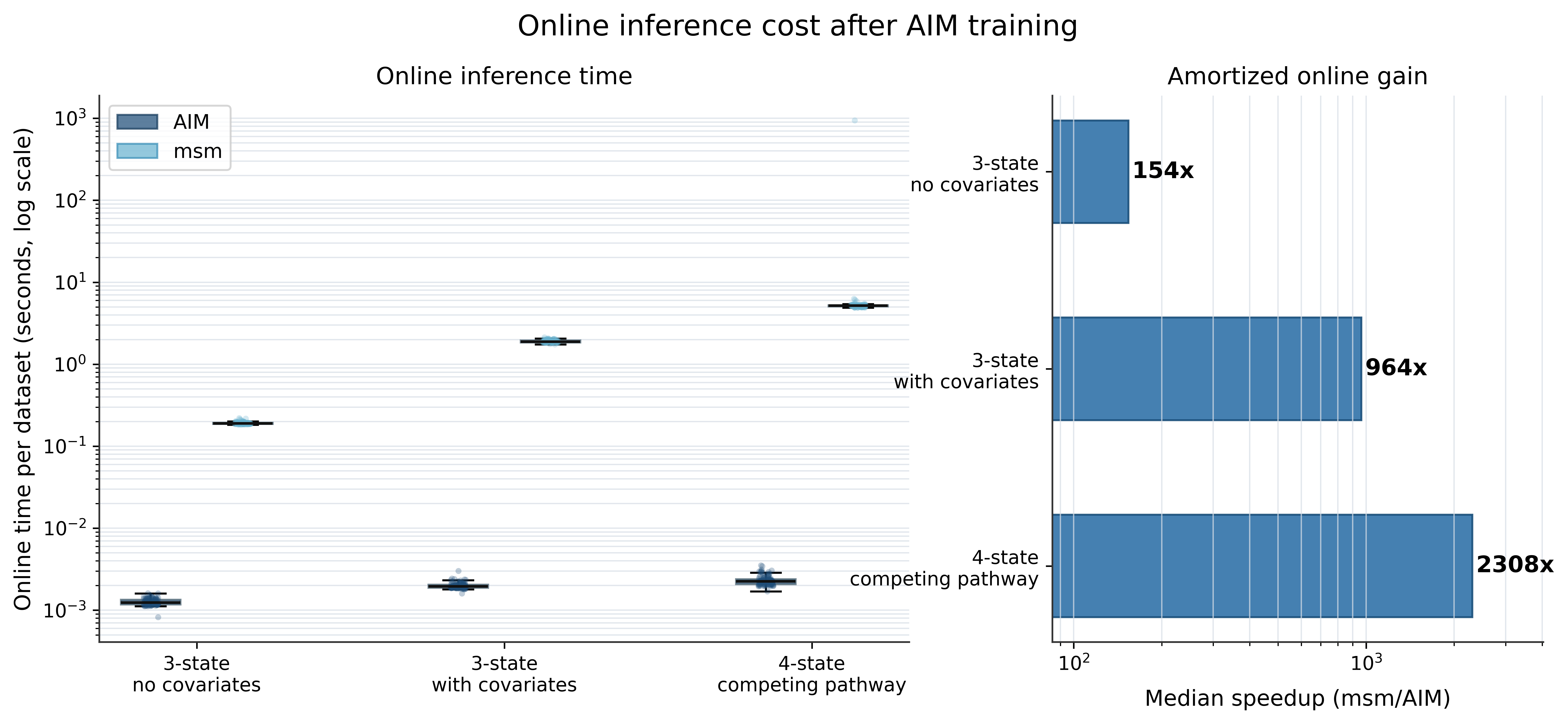}
\caption{
Online inference time for AIM and the likelihood-based \texttt{msm} estimator across the simulation scenarios. Times are shown on a logarithmic scale, with boxplots summarizing variation across replicated test datasets. Numbers above each scenario indicate the median speedup, defined as the ratio of \texttt{msm} runtime to AIM runtime. After offline training, AIM performs inference in milliseconds, whereas \texttt{msm} requires a separate likelihood optimization for each dataset.
}
\label{fig:online_time}
\end{figure}

\begin{table}[htbp]
\centering
\caption{Simulation results for the three-state models with and without covariates.}
\label{tab:simulation_3state}
\resizebox{\textwidth}{!}{%
\begin{tabular}{llrrrrrrr}
\toprule
Model & Parameter & True & AIM Mean & msm Mean & AIM SD & msm SD & AIM Cov. & msm Cov. \\
\midrule
3-state no covariates & b12\_0 & -0.600 & -0.605 & -0.601 & 0.022 & 0.022 & 0.980 & 0.980 \\
3-state no covariates & b13\_0 & -1.000 & -0.986 & -1.000 & 0.030 & 0.029 & 0.930 & 0.960 \\
3-state no covariates & b23\_0 & -0.200 & -0.199 & -0.200 & 0.029 & 0.029 & 0.980 & 0.980 \\
3-state with covariates & b12\_0 & -0.600 & -0.590 & -0.601 & 0.027 & 0.027 & 0.960 & 0.910 \\
3-state with covariates & b12\_1 & 0.500 & 0.472 & 0.500 & 0.037 & 0.025 & 0.850 & 0.920 \\
3-state with covariates & b12\_2 & -0.300 & -0.315 & -0.303 & 0.030 & 0.025 & 0.950 & 0.980 \\
3-state with covariates & b13\_0 & -1.000 & -1.001 & -1.004 & 0.038 & 0.036 & 0.970 & 0.910 \\
3-state with covariates & b13\_1 & 0.400 & 0.409 & 0.405 & 0.044 & 0.035 & 0.940 & 0.900 \\
3-state with covariates & b13\_2 & 0.200 & 0.199 & 0.199 & 0.038 & 0.031 & 0.950 & 0.940 \\
3-state with covariates & b23\_0 & -0.200 & -0.243 & -0.204 & 0.037 & 0.035 & 0.870 & 0.950 \\
3-state with covariates & b23\_1 & -0.500 & -0.514 & -0.504 & 0.058 & 0.041 & 0.930 & 0.920 \\
3-state with covariates & b23\_2 & 0.600 & 0.584 & 0.602 & 0.049 & 0.037 & 0.950 & 0.930 \\
\bottomrule
\end{tabular}%
}
\end{table}

\begin{table}[htbp]
\centering
\caption{Simulation results for the four-state competing pathway model with covariates.}
\label{tab:simulation_4state}
\resizebox{\textwidth}{!}{%
\begin{tabular}{llrrrrrrr}
\toprule
Model & Parameter & True & AIM Mean & msm Mean & AIM SD & msm SD & AIM Cov. & msm Cov. \\
\midrule
4-state with covariates & b12\_0 & -0.800 & -0.830 & -0.804 & 0.027 & 0.023 & 0.870 & 0.950 \\
4-state with covariates & b12\_1 & 0.500 & 0.494 & 0.503 & 0.032 & 0.024 & 0.930 & 0.940 \\
4-state with covariates & b12\_2 & -0.300 & -0.334 & -0.299 & 0.028 & 0.022 & 0.870 & 0.970 \\
4-state with covariates & b13\_0 & -1.000 & -1.019 & -0.999 & 0.027 & 0.024 & 0.930 & 0.980 \\
4-state with covariates & b13\_1 & 0.400 & 0.394 & 0.400 & 0.037 & 0.029 & 0.910 & 0.920 \\
4-state with covariates & b13\_2 & 0.200 & 0.199 & 0.200 & 0.034 & 0.026 & 0.960 & 0.920 \\
4-state with covariates & b24\_0 & -0.400 & -0.414 & -0.402 & 0.041 & 0.036 & 0.980 & 0.920 \\
4-state with covariates & b24\_1 & -0.500 & -0.507 & -0.500 & 0.048 & 0.040 & 0.970 & 0.950 \\
4-state with covariates & b24\_2 & 0.600 & 0.603 & 0.602 & 0.050 & 0.041 & 0.900 & 0.950 \\
4-state with covariates & b34\_0 & -0.300 & -0.323 & -0.304 & 0.049 & 0.040 & 0.900 & 0.940 \\
4-state with covariates & b34\_1 & 0.200 & 0.194 & 0.206 & 0.055 & 0.038 & 0.940 & 0.930 \\
4-state with covariates & b34\_2 & 0.500 & 0.496 & 0.504 & 0.059 & 0.038 & 0.910 & 0.910 \\
\bottomrule
\end{tabular}%
}
\end{table}

\subsection{Sensitivity analysis}

To investigate the robustness of the proposed procedure, we examined the effects of both the observed sample size and the number of simulated training datasets.

\paragraph{Sensitivity to sample size.}

We first fixed the number of training simulations at

\[
M=50{,}000,
\]

and considered sample sizes

\[
N\in\{500,\;1250,\;2500,\;5000\}.
\]

For each value of \(N\), one hundred independent datasets were generated under the three simulation scenarios, and posterior summaries were obtained using the trained network. Figure~\ref{fig:sensitivity} (top row) reports the average posterior standard deviations and empirical coverage probabilities of the nominal \(95\%\) credible intervals.

\begin{figure}
\centering
\includegraphics[width=\textwidth]{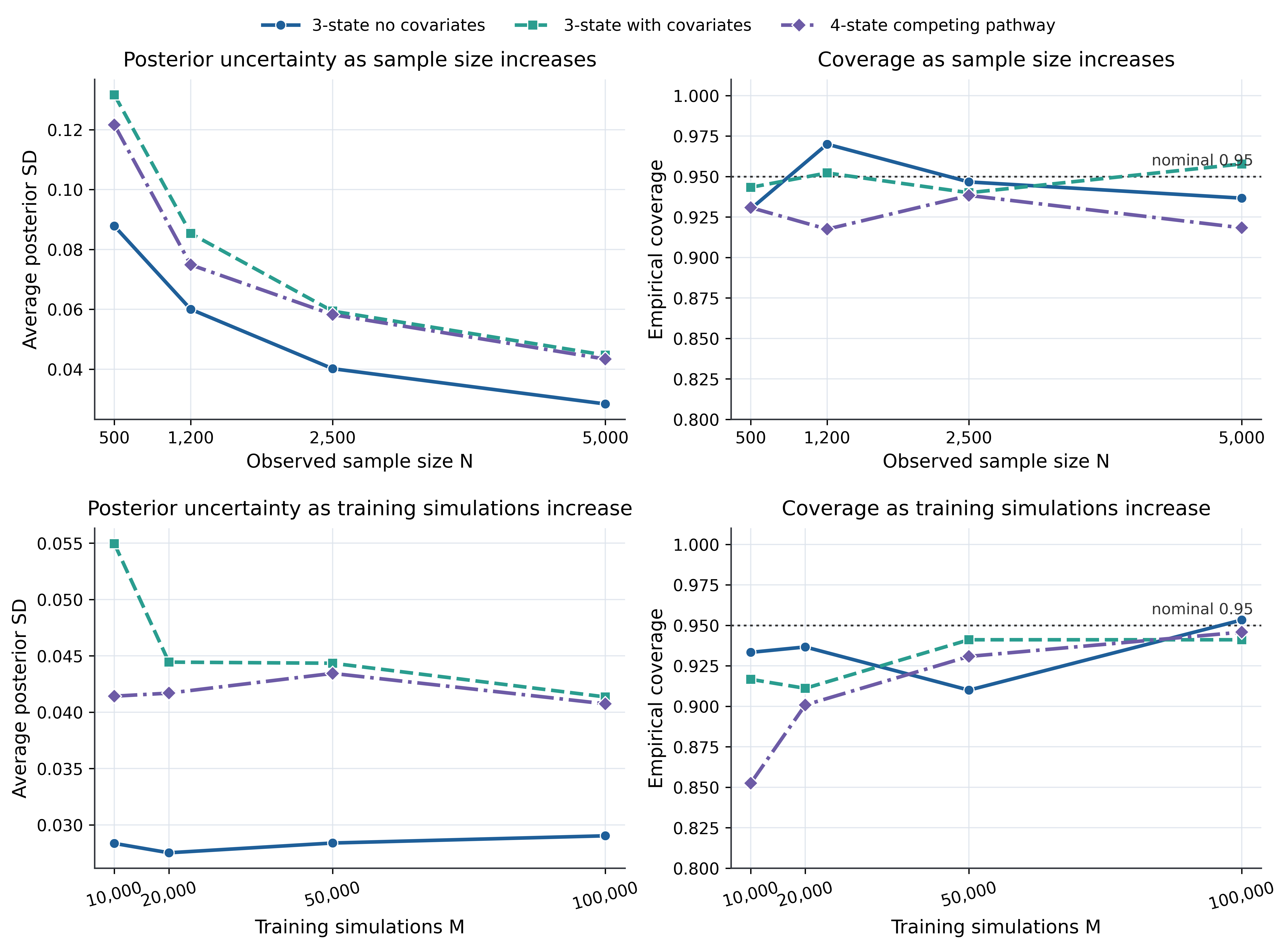}
\caption{
Sensitivity analysis with respect to the observed sample size and the number of simulated training datasets. The top row fixes the number of training simulations at \(M=50{,}000\) and shows the effect of increasing the sample size \(N\). The bottom row fixes \(N=5000\) and investigates the effect of varying the number of training simulations \(M\). Left panels show average posterior standard deviations, while right panels show empirical coverage probabilities for nominal \(95\%\) credible intervals. The horizontal dotted line indicates the nominal coverage level.
}
\label{fig:sensitivity}
\end{figure}

Across all three scenarios, posterior uncertainty decreased substantially as the sample size increased. In particular, the average posterior standard deviation exhibited approximately monotone decay, reflecting the increasing information available in larger datasets. Despite the reduction in posterior variability, empirical coverage probabilities remained stable and stayed close to the nominal level of \(95\%\), indicating that uncertainty quantification was well calibrated over a wide range of sample sizes.

\paragraph{Sensitivity to the number of training simulations.}

Next, we fixed the observed sample size at

\[
N=5000
\]

and varied the number of simulated training datasets,

\[
M\in
\{10{,}000,\;20{,}000,\;50{,}000,\;100{,}000\}.
\]

Figure~\ref{fig:sensitivity} (bottom row) summarizes the corresponding posterior standard deviations and empirical coverage probabilities.

The posterior standard deviations were relatively insensitive to the number of training simulations once

\[
M\ge 20{,}000,
\]

suggesting that the learned posterior approximation had largely stabilized. Similarly, empirical coverage probabilities gradually approached the nominal level as the training size increased, with only minor changes beyond

\[
M=50{,}000.
\]

These results indicate that the proposed amortized inference procedure is robust to the choice of training size and that a moderate number of simulated datasets is sufficient to achieve reliable posterior inference.

Overall, the sensitivity analysis demonstrates two complementary properties of the proposed framework. Increasing the observed sample size primarily improves statistical efficiency, whereas increasing the number of training simulations mainly reduces approximation error in the learned posterior distribution. Together, these findings are consistent with the asymptotic results established in Section~\ref{sec:theory}. Detailed numerical summaries corresponding to the sensitivity analyses are provided in Tables S1--S2 of the Supplementary Material.

\section{Application}
\label{sec:application}

We applied AIM to the CAV panel data \citep{sharples2003diagnostic} available in the
\texttt{msm} package \citep{jackson2011multi}, comprising 622 heart-transplant
recipients observed on a yearly grid over a ten-year window across four states
(CAV-free, mild CAV, moderate/severe CAV, and death). Reverse transitions, which
reflect misclassification of an irreversible process, were recoded upward, and
observations were aligned to a common fixed schedule so that the interval-count
representation of Section~\ref{sec:method} applies. We fitted the Markov model
with allowed transitions $1\!\to\!2$, $1\!\to\!4$, $2\!\to\!3$, $2\!\to\!4$, and
$3\!\to\!4$, adjusting for a binary donor-age covariate (\texttt{dage\_hi}). The
same rectangular panel was used to obtain maximum-likelihood estimates from
\texttt{msm}, giving a like-for-like comparison.

Table~\ref{tab:cav_aim_vs_msm} shows that AIM produced estimates that were highly concordant with the likelihood-based \texttt{msm} analysis. The baseline log-intensities were close across the two methods for all five transitions. In particular, the AIM estimates for $b_{12,0}$, $b_{14,0}$, $b_{23,0}$, and $b_{34,0}$ differed from the corresponding \texttt{msm} estimates by less than 0.10 on the log-intensity scale. The largest baseline difference occurred for the sparsely observed $2\to4$ transition ($b_{24,0}$), for which AIM estimated $-3.295$ compared with the \texttt{msm} estimate of $-2.920$.

The donor-age effects also showed strong agreement in direction. Both methods estimated a positive donor-age effect for the $1\to2$ transition and negative effects for the remaining four transitions. The AIM estimates for $b_{12,1}$, $b_{14,1}$, and $b_{23,1}$ were particularly close to the \texttt{msm} estimates, while larger differences appeared for the less frequent downstream transitions $2\to4$ and $3\to4$. Overall, every \texttt{msm} point estimate fell within the corresponding AIM 95\% credible interval, supporting the calibration of the amortized posterior approximation in this real-data application.

Crucially, once trained, AIM returns the full posterior for a new panel in
milliseconds, without any per-dataset likelihood optimisation or matrix
exponentiation, whereas each \texttt{msm} fit re-optimises the panel likelihood
from scratch. On this dataset the two therefore yield concordant inference, but
only AIM amortises its cost across repeated analyses.

\begin{table}[htbp]
\centering
\caption{Real-data CAV analysis. AIM posterior summaries are compared with
likelihood-based estimates from the \texttt{msm} package under the same
four-state transition structure and identical fixed-schedule panel. Parameters
with subscript 0 are transition intercepts (log-intensities at the covariate
baseline); parameters with subscript 1 are log-hazard-ratio coefficients for
\texttt{dage\_hi} (donor age above the median). $n$ is the observed number of
transitions of each type; ``At risk'' is the number of subjects who ever occupy
the origin state.}
\label{tab:cav_aim_vs_msm}
\begin{tabular}{lrrrrrr}
\toprule
Parameter & $n$ & At risk & AIM estimate & AIM 95\% CI & MSM estimate & MSM 95\% CI \\
\midrule
$b_{12,0}$ & 178 & 622 & $-3.049$ & $[-3.232,\,-2.857]$  & $-3.121$ & $[-3.330,\,-2.917]$ \\
$b_{12,1}$ &  &  & $0.245$ & $[0.025,\,0.453]$  & $0.223$ & $[-0.048,\,0.494]$ \\
$b_{14,0}$ & 78 & 622 & $-4.090$ & $[-4.438,\,-3.761]$  & $-4.013$ & $[-4.358,\,-3.666]$ \\
$b_{14,1}$ &  &  & $-0.188$ & $[-0.637,\,0.272]$  & $-0.179$ & $[-0.655,\,0.298]$ \\
$b_{23,0}$ & 50 & 178 & $-2.008$ & $[-2.283,\,-1.740]$  & $-1.974$ & $[-2.285,\,-1.655]$ \\
$b_{23,1}$ &  &  & $-0.199$ & $[-0.543,\,0.151]$  & $-0.170$ & $[-0.600,\,0.260]$ \\
$b_{24,0}$ & 29 & 178 & $-3.295$ & $[-3.895,\,-2.767]$  & $-2.920$ & $[-3.557,\,-2.259]$ \\
$b_{24,1}$ &  &  & $-0.148$ & $[-0.774,\,0.476]$  & $-0.466$ & $[-1.298,\,0.366]$ \\
$b_{34,0}$ & 28 & 88 & $-2.158$ & $[-2.533,\,-1.776]$  & $-2.082$ & $[-2.581,\,-1.554]$ \\
$b_{34,1}$ &  &  & $-0.516$ & $[-1.033,\,-0.002]$  & $-0.748$ & $[-1.489,\,-0.008]$ \\
\bottomrule
\end{tabular}
\end{table}

\section{Discussion}

We have proposed AIM, an amortized neural Bayes procedure for
interval-censored continuous time MSTMs. The central idea is simple:
simulate extensively from a specified multistate model class, learn the mapping
from likelihood-informed summaries to parameters, and reuse the learned
posterior approximation for future datasets. This changes the computational
profile of multistate inference. Instead of solving a new likelihood
optimization problem for every panel, AIM performs the expensive work offline
and reduces online inference to summary construction and a neural-network
evaluation.

The methodological contribution is to make this amortization compatible with
classical multistate structure. The summaries are not generic black-box
features; they are built from interval transition counts, state occupancies, and
covariate-stratum proportions, which correspond directly to the observed panel
likelihood. This connection supports the theoretical analysis. Under fixed
visit schedules and discretized covariate strata, the transition count tables
are sufficient for the panel likelihood. With reachability and full-rank design
conditions, the population summaries identify the transition-intensity
parameters. These ingredients yield consistency of the summary posterior, and a
sufficiently accurate mixture-density approximation inherits that consistency.

The CAV analysis illustrates both the promise and the current scope of the
proposed framework. Under a fixed observation schedule and a discretised
covariate, AIM recovers essentially the same transition-intensity estimates as
established likelihood-based software, providing external validation that the
amortised posterior is well calibrated on real data rather than only in
simulation. The one point of disagreement, the marginal donor-age effect on the
sparsely observed $3\!\to\!4$ transition, does not reflect a systematic bias but
rather the intrinsic difficulty of estimating covariate effects from very few
events, a difficulty shared by any method. Notably, AIM additionally supplies
covariate-adjusted inference that the sequential ABC treatment of the same data
\citep{tancredi2019approximate} did not attempt, while retaining full agreement with
\texttt{msm} on the shared quantities.

Two limitations specific to this application deserve emphasis. First, aligning
the irregularly observed transplant follow-up to a common yearly grid discards
some timing information; although this mirrors standard practice for discretely
observed multistate data, the sensitivity of conclusions to the discretisation
window warrants further study. Second, the CAV data are modest in size and their
covariates are only weakly informative, so this application is best read as a
concordance and calibration check against a gold-standard implementation rather
than as a source of new clinical findings. The principal advantage of amortisation is the ability to reuse a single trained network across many datasets, subgroups, or resampling-based analyses without repeatedly evaluating the likelihood. This advantage is most consequential 
when inference must be performed many times, a setting that is not fully represented 
by the present single-cohort analysis.
The simulation results clarify the trade-off. Direct likelihood fitting through
\texttt{msm} remains a strong benchmark under correct specification and, in our
experiments, produced slightly better interval calibration. AIM, however,
delivered similar point estimates and broadly reasonable uncertainty
quantification while reducing online inference time by one to two orders of
magnitude or more. This is the relevant operating regime for amortized
inference: not replacing a single carefully tuned likelihood fit, but enabling
large numbers of repeated analyses after a model class has been fixed.

Several limitations remain. First, the present theory and implementation assume
a common observation schedule and discretized covariate strata. These
assumptions make the summary map fixed-dimensional and interpretable, but they
are restrictive for studies with irregular follow-up or continuous covariate
effects that one does not want to discretize. Second, AIM is currently trained
for a specified state space and transition graph. A different graph generally
requires a new simulator, summary map, and trained network. Third, the summary
statistics are hand designed. Their connection to the likelihood is a strength,
but learned encoders based on sequences, graphs, or uniformization
representations may recover additional information in settings where the present
summaries are no longer sufficient. Finally, the mild undercoverage observed in
some simulations indicates that calibration diagnostics should accompany any
new deployment, especially when training budgets are limited or the observed
dataset lies near the edge of the prior predictive distribution.

Future work should therefore focus on irregular observation processes,
continuous covariate representations, calibration adjustments for neural
posteriors, and architectures that can share information across related
transition graphs. The broader message is that amortized inference can be
integrated with the structure of classical multistate models rather than used as
a purely black-box substitute. That integration is what makes reusable Bayesian
inference for multistate transition analysis a realistic target.

\section*{Data Availability Statement}

The cardiac allograft vasculopathy data analyzed in this study are publicly available as the \texttt{cav} dataset in the R package \texttt{msm}. The dataset can be loaded in R using \texttt{data("cav", package = "msm")}.

\appendix
\renewcommand{\thetable}{S\arabic{table}}
\renewcommand{\theHtable}{S\arabic{table}}
\setcounter{table}{0}
\section{Proof of Lemma~\ref{lem:sufficiency_counts}}
\label{app:sufficiency_proof}

Fix a stratum \(g\). For subject \(i\in g\), write the observed panel trajectory as
\[
\mathbf x_i=(x_i^{(0)},x_i^{(1)},\ldots,x_i^{(T)}),
\]
where \(x_i^{(k)}=X_i(t_k)\). By the Markov property of the continuous-time process observed at the fixed visit times,
\[
\Pr_\theta(\mathbf X_i=\mathbf x_i\mid i\in g)
=
\Pr\{X_i(t_0)=x_i^{(0)}\}
\prod_{k=1}^T
p_{x_i^{(k-1)}x_i^{(k)}}^{(k,g)}(\theta).
\]
The first factor does not depend on \(\theta\). Conditional independence across subjects gives
\[
L_g(\theta;X)
=
\prod_{i\in g}\Pr\{X_i(t_0)=x_i^{(0)}\}
\prod_{i\in g}\prod_{k=1}^T
p_{x_i^{(k-1)}x_i^{(k)}}^{(k,g)}(\theta).
\]
Grouping identical interval transitions yields
\[
\prod_{i\in g}\prod_{k=1}^T
p_{x_i^{(k-1)}x_i^{(k)}}^{(k,g)}(\theta)
=
\prod_{k=1}^T\prod_{r=1}^S\prod_{s=1}^S
\{p_{rs}^{(k,g)}(\theta)\}^{n_{rs}^{(k,g)}}.
\]
Therefore,
\[
L_g(\theta;X)
=
h_g(X)g_g(T_g(X),\theta),
\]
where
\[
h_g(X)=\prod_{i\in g}\Pr\{X_i(t_0)=x_i^{(0)}\}
\]
does not depend on \(\theta\), and
\[
g_g(T_g(X),\theta)
=
\prod_{k,r,s}
\{p_{rs}^{(k,g)}(\theta)\}^{n_{rs}^{(k,g)}}.
\]
Across strata,
\[
L(\theta;X)=\prod_{g=1}^G L_g(\theta;X)
=
\left\{\prod_g h_g(X)\right\}
\left\{\prod_g g_g(T_g(X),\theta)\right\}.
\]
The first factor is free of \(\theta\), and the second depends on the data only through the pooled count tables \(T(X)\). By the Fisher--Neyman factorization theorem, \(T(X)\) is sufficient for \(\theta\) under the discretized-covariate observed panel likelihood.

\section{Additional results for sensitivity analyses} Detailed numerical summaries corresponding to the sensitivity analyses reported in Figure~\ref{fig:sensitivity} are presented below. 

\begin{table}[!htbp]
\caption{Sensitivity to sample size. Number of training simulations fixed at $M=50000$; performance averaged over model parameters and 100 test datasets.}
\label{tab:sens_sampleN}
\begin{tabular}{lrrrrr}
\toprule
Scenario & N & Post. SD & RMSE & Bias & Coverage \\
\midrule
3-state no covariates & 500 & 0.088 & 0.092 & 0.014 & 0.930 \\
3-state no covariates & 1200 & 0.060 & 0.057 & 0.006 & 0.970 \\
3-state no covariates & 2500 & 0.040 & 0.041 & 0.005 & 0.947 \\
3-state no covariates & 5000 & 0.028 & 0.028 & 0.003 & 0.937 \\
3-state with covariates & 500 & 0.132 & 0.132 & 0.035 & 0.943 \\
3-state with covariates & 1200 & 0.085 & 0.083 & 0.013 & 0.952 \\
3-state with covariates & 2500 & 0.059 & 0.060 & 0.018 & 0.940 \\
3-state with covariates & 5000 & 0.045 & 0.043 & 0.010 & 0.958 \\
4-state with covariates & 500 & 0.122 & 0.128 & 0.030 & 0.931 \\
4-state with covariates & 1200 & 0.075 & 0.085 & 0.016 & 0.918 \\
4-state with covariates & 2500 & 0.058 & 0.060 & 0.013 & 0.938 \\
4-state with covariates & 5000 & 0.043 & 0.048 & 0.021 & 0.918 \\
\bottomrule
\end{tabular}
\end{table}

\begin{table}[!htbp]
\caption{Sensitivity to the number of training simulations. Sample size fixed at $N=5000$; performance averaged over model parameters and 100 test datasets.}
\label{tab:sens_trainM}
\begin{tabular}{lrrrrr}
\toprule
Scenario & M & Post. SD & RMSE & Bias & Coverage \\
\midrule
3-state no covariates & 10000 & 0.028 & 0.029 & 0.005 & 0.933 \\
3-state no covariates & 20000 & 0.028 & 0.030 & 0.008 & 0.937 \\
3-state no covariates & 50000 & 0.028 & 0.030 & 0.007 & 0.910 \\
3-state no covariates & 100000 & 0.029 & 0.028 & 0.003 & 0.953 \\
3-state with covariates & 10000 & 0.055 & 0.054 & 0.022 & 0.917 \\
3-state with covariates & 20000 & 0.044 & 0.049 & 0.018 & 0.911 \\
3-state with covariates & 50000 & 0.044 & 0.046 & 0.011 & 0.941 \\
3-state with covariates & 100000 & 0.041 & 0.043 & 0.013 & 0.941 \\
4-state with covariates & 10000 & 0.041 & 0.052 & 0.026 & 0.853 \\
4-state with covariates & 20000 & 0.042 & 0.047 & 0.015 & 0.901 \\
4-state with covariates & 50000 & 0.043 & 0.047 & 0.018 & 0.931 \\
4-state with covariates & 100000 & 0.041 & 0.042 & 0.010 & 0.946 \\
\bottomrule
\end{tabular}
\end{table}

\clearpage
\bibliography{references}
\end{document}